\documentclass[a4paper,12pt]{article}
\usepackage{graphicx}
\usepackage{xcolor} % for \textcolor markup of edits; loaded before hyperref
\usepackage[authoryear]{natbib}%%% \citet{}, \citep{}, ...
\usepackage{fullpage}
\graphicspath{{Figures/}} %%% Allows to use file names directly without specifying path in
\usepackage{booktabs,array,longtable}
\usepackage{amsmath}
\usepackage{amsfonts}
\usepackage{amssymb}
\usepackage{threeparttable} % for tablenotes environment
\usepackage{amsthm}     % for using \theoremstyle{plain}, {definition}, {remark}
\usepackage{algorithm} % for floating algorithm environment
\usepackage{tabularx}
\theoremstyle{plain}   % default. it sets the text in italic and adds extra space above and below the
\newtheorem{theorem}{Theorem}

\newtheorem{corollary}{Corollary}

\newtheorem{lemma}{Lemma}

\theoremstyle{definition} % adds extra space above and below, but sets the text in roman.
\theoremstyle{remark} % Font is set in roman, with no additional space above or below.
\newtheorem{remark}{Remark}

\newcommand{\Mode}{\textup{Mode\,}}
\renewenvironment{proof}[1][Proof]{\noindent\emph{#1.} }{\ $\Box$\medskip}
\usepackage{enumitem} % for control of vertical space in \itermize environment
\setlist[itemize]{itemsep=1pt, topsep=4pt}
\usepackage{titlesec} % for control of vertical space in \paragraph
\titlespacing*{\paragraph}
{0pt} % left indent
{1.5ex} % space before
{0.5em} % space after (horizontal, since \paragraph is run-in)
\usepackage[bookmarks=false, colorlinks=true, citecolor=blue, linkcolor=blue, urlcolor=blue]{hyperref}
\usepackage{cleveref}
\begin{document}

\title{Scheduling Jobs with Multiple Operational Modes \\ and Tail Times}

\author{Bo Chen\footnote{Warwick Business School, University of Warwick, Coventry, CV4 7AL, UK\@;
        b.chen@warwick.ac.uk }
        \and Jelmer Pier van der Gaast\footnote{School of Management, Fudan University, Shanghai, 200433, China; jpgaast@fudan.edu.cn }
        \and Xiandong Zhang\footnote{School of Management, Fudan University, Shanghai, 200433, China; xiandongzhang@fudan.edu.cn }}

\date{July 19, 2026}

% Jelmer:
% All my edits are in blue. Let me know if everything is okay, then I will remove them.
% Search of ``Jelmer" in text for some additional points that need attention.

\maketitle

\begin{abstract}
\noindent This study explores a scheduling challenge inspired by the production of programmable materials, such as advanced liquid crystal displays. In these systems, the final quality of a product is reached only after a resource-free maturation period, known as a ``tail", during which the machine is available for processing other jobs. Each job can be executed in one of several operational modes, with each mode determining a specific combination of machine processing time and subsequent tail duration. The primary task is to simultaneously choose the best mode for every job and determine their processing order. We analyze this model across several key performance goals, including the total time required to finish all jobs, the synchronization of completion times (the gap between the earliest and latest finished products), and the total weighted completion time. Our findings provide a detailed classification of the computational complexity of these problems. We demonstrate that while traditional versions with only one mode per job are simple to solve using standard rules, the introduction of just two modes makes finding optimal solutions for most of these goals computationally difficult. When the number of available modes is large, the complexity increases significantly. However, we also identify specific scenarios that remain efficiently solvable, such as when the processing order is already determined or when the goal is to minimize the average completion time. These results offer theoretical clarity and practical strategies for optimizing complex manufacturing and chemical processes involving forced cooling or maturation stages.

\smallskip\noindent
\textbf{Keywords:} scheduling, operational mode, tail time, NP-hardness, computational complexity
\end{abstract}

\section{Introduction}\label{sec:introduction}

Recent advances in programmable materials have enabled printing processes in which final product
quality, such as color in blue-phase liquid crystal (BPLC) displays, emerges not at the moment of
deposition, but during a subsequent, resource-free maturation phase. In the BPLC printing system
described by \cite{Yang2024}, a single machine deposits ink onto a polymer film. Each printed cell
occupies the machine for a certain processing time, after which the cell undergoes a
diffusion-driven post-processing stage, referred to as the \emph{tail}, during which the machine is
free to process other cells. The tail duration determines the final color and depends on the chosen
printing resolution, which involves parameters such as the number of layers or the droplet spacing.
The practical challenge is twofold: first, for each cell, one must select among several available
operational modes, each defining a processing-time and tail-duration pair; second, one must
sequence all cells on the single machine so that the moments at which the cells reach their target
colors are as synchronized as possible. {Synchronization matters because diffusion does not stop
once a cell attains its target color: the printed film is stabilized in a single post-printing
step, so cells whose maturation completes much earlier or later than that of the others are fixed
past or short of their intended colors, degrading the uniformity of the display.} This
synchronization requirement naturally leads to the objective of minimizing the spread between the
earliest and latest completion times, which is one of several scheduling criteria studied in this
paper.

The same scheduling structure also arises in prefabricated concrete production \citep{MA2026100172}. In a prefabrication plant, each concrete component must first undergo casting operations (such as mixing, pouring, and vibration) on a single machine. Once casting is completed, the component enters a curing stage during which it gains sufficient strength for demoulding. Because curing does not require the casting machine, production of subsequent components can begin immediately. The curing time, however, depends strongly on the chosen concrete mix. High-early-strength formulations, for example those with additional cement or chemical accelerators, can substantially shorten curing but typically require more intensive processing during casting. By contrast, slower-setting mixes reduce machine processing requirements while extending the off-machine curing period. Each mix therefore corresponds to an operational mode characterized by a specific processing time and tail time. The planner must determine both the mix design for each component and the processing sequence on the machine, with objectives such as meeting delivery schedules, reducing overall production time, or coordinating the completion of multiple components. Despite the different application context, the resulting optimization problem is mathematically equivalent to the BPLC printing setting described above.

These practical settings give rise to a fundamental scheduling model. In this model, each job, or
cell, can be processed in one of several modes; each mode specifies a processing time on the
machine and a tail time that does not occupy the machine. The decision variables thus consist of
both a mode assignment and a job permutation. {The objective is to minimize the maximum job
completion time (i.e., the makespan), the difference between the maximum and minimum job completion
times, or the total weighted job completion time, or alternatively to maximize the minimum job
completion time.} We give a precise problem definition in \Cref{sec:preliminaries}.

Although the special case with a single mode per job is well understood, the introduction of
multiple operational modes fundamentally changes the nature of the problem. When the job sequence
is fixed, mode selection alone can be solved in polynomial time for several objectives. {However,
when the mode assignment and the job sequence must be determined jointly, the scheduling problems
become significantly harder.}

The main contributions of this study are as follows. We formally define the problem model and establish basic optimality properties for the single-mode case in \Cref{sec:preliminaries}. We then provide a complete complexity classification for all basic variants of the model. \Cref{tbl:main-results,tbl:complexity} at the end of \Cref{sec:preliminaries} summarize which variants are polynomially solvable, which are weakly NP-hard, and which are strongly NP-hard. % Jelmer: I do not really like stating our contribution as a forward reference. However otherwise we need to explain quite some notation already.

The remainder of this paper is organized as follows. \Cref{sec:literature-review} reviews the related literature on scheduling with tails and with multiple modes. \Cref{sec:preliminaries} presents the full problem description, introduces the relevant notation, states fundamental lemmas, and provides summary tables for the main results and the complexity landscape. \Cref{sec:weakly-np-hard,sec:strongly-np-hard} establish the weak and strong NP-hardness results, respectively. \Cref{sec:conclusions} concludes with a discussion of practical implications and directions for future research.

\section{Literature review}\label{sec:literature-review}

Throughout this review, we use the standard three-field notation $\alpha | \beta | \gamma$ of
\cite{graham1979} to describe scheduling problems, where $\alpha$ specifies the machine
environment, $\beta$ specifies job characteristics and constraints (e.g., $q_j$ for delivery or tail times), and $\gamma$ specifies the objective function to be minimized or maximized (e.g., $C_{\max}$ for makespan, $\sum w_j C_j$ for total weighted completion time). We adopt this notation throughout the paper, including for the multi-mode model introduced in \Cref{sec:preliminaries}, where it is extended to accommodate mode-dependent processing and tail times.

This paper draws on two largely separate streams of scheduling research. The first concerns single-machine models in which a resource-free interval delays a job's completion after processing ends, which is variously termed delivery times, tails, or delivery lead times, a paradigm that captures settings where a finished item must still travel, cool, cure, or otherwise settle before it counts as complete. The second concerns models in which a job's processing time itself is a decision variable, chosen from a discrete or continuous set at some cost or under some constraint, rather than a fixed input. Our model sits at the intersection of these two traditions: each job's processing time and its post-processing tail are jointly determined by a discrete mode choice, so that selecting a mode simultaneously fixes how long the job occupies the machine and how long it continues to ``finish'' afterward. \Cref{sec:delivery-time} reviews the literature on scheduling with delivery times, and \Cref{sec:controllable-time} reviews work on controllable processing times, before \Cref{sec:connection} positions our contribution relative to both.

\subsection{Scheduling with delivery times}\label{sec:delivery-time}

Scheduling with delivery times, commonly referred to in the literature as scheduling with tails, is a classical paradigm in deterministic scheduling theory. In this model, each job $j$ is characterized by a processing time $p_j$ on the machine and a delivery time $q_j$ (or tail) that commences upon completion of processing, consumes no machine resource, and may proceed in parallel with subsequent jobs. The completion time of job $j$ is defined as $C_j = S_j + p_j + q_j$, where $S_j$ denotes its start time.

\paragraph{Foundational results.} To the best of our knowledge,
\cite{potts1980analysis} was the first to explicitly address a model integrating both production
scheduling and subsequent delivery operations on a single machine. In his formulation, each job $j$ is specified by a release date $r_j$, a processing time $p_j$, and a delivery time $q_j$. A central observation in \cite{potts1980analysis} is that the problem admits a symmetric transformation to the classical maximum-lateness formulation: by setting due dates $d_j = K - q_j$ for a sufficiently large constant $K$, minimizing the makespan becomes equivalent to minimizing the maximum lateness $L_{\max}$ with respect to these due dates. The computational complexity of the single-machine problem was established by \cite{lenstra1977complexity}, who proved that $1|r_j|L_{\max}$ (equivalently, $1|r_j,q_j|C_{\max}$) is NP-hard. When all release dates are zero, the problem reduces to $1|q_j|C_{\max}$, which is solvable in $O(n \log n)$ time by sequencing jobs in non-increasing order of $q_j$. This result is commonly referred to as Jackson's rule~\citep{jacksons1955}.

\paragraph{Approximation algorithms and complexity refinements.}
\cite{hall1992jacksons} developed a polynomial-time approximation scheme (PTAS) for
$1|r_j,q_j|C_{\max}$ and also provided a $(4/3)$-approximation algorithm for the variant with
precedence constraints\@. \cite{NOWICKI199469} proposed a $(3/2)$-approximation algorithm for the
same problem that runs in $O(n \log n)$ time. For the parallel-machine setting,
\cite{woeginger1994heuristics} analyzed $P|q_j|C_{\max}$ and established a $(2 -
{2}/({m+1}))$-approximation for the offline case, where $m$ is the number of machines, and a
$2$-approximation for the online case\@. \cite{hall1992approximation} further provided a PTAS for
$P|r_j,q_j|C_{\max}$\@. \cite{Hoogeveen1990} studied the problem $1||f(E_{\max},L_{\max})$, {where
$f(\cdot, \cdot)$ is a function of the maximum earliness $E_{\max}$ and the maximum lateness
$L_{\max}$ that is non-decreasing in both arguments}. He presented $O(n^2 \log n)$ algorithms for
the no-idle variant and proved strong NP-hardness otherwise. This result implies that
$1|q_j|D_{\max}$ (where $D_{\max}=C_{\max}-C_{\min}$) is polynomially solvable when idle time is
prohibited but becomes strongly NP-hard when idle time is permitted.

\paragraph{Integrated production and outbound distribution.}
\cite{Chen2010,chen2025} comprehensively surveyed studies under the broader context of integrated production and outbound distribution, which aim to jointly optimize production scheduling and transportation operations, offline or online. Within this framework, \cite{ChengKahlbacher1993} examined batch delivery models and showed that the problem of simultaneously determining the optimal number of batches and the job sequence to minimize the sum of delivery and waiting costs is NP-hard\@. \cite{Cheng1996} further established connections between batch delivery problems and parallel machine scheduling, and developed polynomial-time algorithms for several special cases\@. \cite{Lee20013} considered a single machine with a finite fleet of capacity-constrained trucks and provided a complexity classification for a range of scheduling problems with transportation considerations.

\paragraph{Multiple shipping modes and sequence-dependent tails.}
\cite{Stecke2007} investigated a production-transportation problem with multiple shipping modes (e.g., overnight, one-day, two-day delivery), where shipping cost varies with transportation lead time, and proved that minimizing total shipping cost is NP-hard\@. \cite{JIA2026410} studied a related model with multiple shipping modes, each characterized by a guaranteed shipping time and a convexly non-increasing shipping cost function; they demonstrated that the feasibility problem is NP-hard for a fixed planning horizon and strongly NP-hard for an arbitrary horizon. A distinct but conceptually related stream of research concerns past-sequence-dependent delivery times (psddt), in which the delivery time of a job depends on its waiting time, i.e., the cumulative processing time of jobs already scheduled\@. \cite{KOULAMAS2010264} proved that the single-machine problem $1|\textup{psddt}|\gamma$ is polynomially solvable for $\gamma \in \{C_{\max}, \sum C_j, L_{\max}, T_{\max}, \sum U_j\}$\@. \cite{wang2021research} further established that the total (discounted) weighted completion time minimization problem in the psddt model is solvable in $O(n \log n)$ time. % Jelmer: Koulamas does not use psddt in their paper but $q_{psd}$

\subsection{Scheduling with controllable processing times}\label{sec:controllable-time}

In classical deterministic scheduling theory, job processing times are typically regarded as fixed, exogenous parameters. However, in numerous practical manufacturing and service systems, processing times can be actively controlled through the allocation of additional resources (such as manpower, overtime, energy, fuel, catalysts, or subcontracting) to job operations \citep{SHABTAY20071643}. In such systems, the scheduler must combine job sequencing and resource allocation decisions to achieve efficient system performance.

The study of scheduling with controllable processing times was initiated by \citet{Vickson01091980, vickson1980} and has since been substantially developed and expanded upon. Early surveys of results up to 1990 were provided by \cite{NOWICKI1990271}, with subsequent surveys by \cite{SHABTAY20071643} and \cite{HOOGEVEEN2005592} in the context of multi-criteria scheduling.

The fundamental insight is that processing time reduction with additional resources typically incurs a cost, creating a trade-off between scheduling performance $F_1$ (e.g., makespan, total completion time, lateness) and resource consumption cost $F_2$. This trade-off gives rise to four prototypical problem formulations:

\begin{itemize}
    \item[$P1$:] Minimize the total integrated cost $f(F_1, F_2)$;
    \item[$P2$:] Minimize $F_1$ subject to $F_2 \le \lambda_1$;
    \item[$P3$:] Minimize $F_2$ subject to $F_1 \le \lambda_2$;
    \item[$P4$:] Identify the set of Pareto-optimal schedules for $(F_1, F_2)$, where $f$ is a linear or convex function and $\lambda_1, \lambda_2$ are fixed values.
\end{itemize}

The prevailing assumption in these studies is that processing times can be continuously varied
within a specified interval by allocating a divisible resource (e.g., electricity, gas)\@.
\cite{CHEN199769} introduced an alternative model in which processing times are discretely
controllable. {In their work, they examined} seven single-machine scheduling problems under this
discrete model, each characterized by an objective function of the form ``scheduling criterion $ +
$ TPC (total processing cost)''. Specifically, for the problem $1|\textup{dm}|\sum C_j +
\text{TPC}$, where ``dm'' stands for discretely controllable processing times, they demonstrated
that it can be reduced to an $n \times n$ assignment problem, which can be solved in $O(n^3)$ time.

More recently, \cite{WANG2026105204} studied the scheduling problem of periodic jobs with discretely controllable processing times on two parallel machines to maximize the overall accumulated utility and developed a pseudo-polynomial time algorithm.

\subsection{Connection to the present work}\label{sec:connection}

The two streams of literature reviewed above address complementary but distinct sources of scheduling flexibility. Work on delivery times treats the tail as an exogenous, fixed attribute of a job, decoupled from any resource decision. Work on controllable processing times, conversely, treats the processing time itself as a decision variable, but almost universally ties that decision to a cost to be traded off against scheduling performance, with no corresponding tail or delivery stage after processing ends. Neither stream captures a setting in which a single discrete choice simultaneously determines both how long a job occupies the machine and how long it continues to mature, cool, or otherwise finish afterward.

The present paper departs from both paradigms by introducing a multi-mode setting in which each job can be processed in one of several alternative modes, with each mode specifying a distinct \emph{pair} of processing time and tail time. The scheduler must simultaneously determine the mode
assignment and the job sequence. This differs from the delivery-time literature in that the tail is no longer fixed but is instead coupled to the processing time through mode selection. On the other hand, it differs from the controllable-processing-time literature in that the ``control'' exercised over the processing time is not purchased at a separate monetary cost but instead manifests directly as a change in the job's tail duration, and hence in its completion time. It also differs from past-sequence-dependent delivery time (psddt) models, where the tail depends on the job's position in the sequence: in our model the tail is mode-dependent rather than position-dependent, and is inherently linked to the processing time rather than to the schedule built around it.

This coupling fundamentally alters the complexity of the resulting problems. As we establish in
\Cref{sec:weakly-np-hard,sec:strongly-np-hard}, objectives that are polynomially solvable in the
classical single-mode setting, where mode selection reduces to a fixed pair $(p_j, q_j )$ per job
$j$, as in \Cref{lem:d1Cmax,lem:d1wjcj}, become NP-hard, and in several cases strongly NP-hard,
once even two operational modes per job are introduced. {This stands in contrast to the
controllable-processing-time literature, where discretely controllable models in which choosing a
mode fixes the processing time and incurs a separate control cost remain polynomially solvable for
the total completion time objective \citep{CHEN199769}, and only the weighted version becomes hard
once a job has two or more modes \citep{cao2006}.} Our results show that when a tail time is
coupled to the mode choice alongside the processing time, hardness can arise even for objectives,
and even at mode counts, where the pure controllable-processing-time analogue remains tractable.
This underscores that it is specifically the joint determination of machine time and
post-processing tail, rather than either flexibility in isolation, that drives the added
computational difficulty.

\section{Problem description and preliminaries}\label{sec:preliminaries}

This section formally develops the scheduling model studied throughout the paper. We begin in \Cref{sec:problem-definition} by defining the problem instance, the decision variables, and the
objective functions of interest, accompanied by \Cref{tbl:main-results,tbl:complexity}, which summarize our main results and the computational complexity landscape of all problem variants considered in this paper. \Cref{sec:single-mode} then presents two preliminary lemmas characterizing the optimal solution to the classical single-mode case, and \Cref{sec:fixed-sequence} establishes polynomial solvability results when the job sequence is fixed. \Cref{sec:single-mode,sec:fixed-sequence} serve both as a baseline against which the general multi-mode results can be compared and as building blocks for several of the hardness proofs in \Cref{sec:weakly-np-hard,sec:strongly-np-hard}.

\subsection{Problem definition}\label{sec:problem-definition}

An instance of our scheduling model is specified by
\begin{equation}\label{eqn:general-instance}
    \mathcal{I}=\left\{(p_{ij},q_{ij})\in
    \mathbb{Z}_{\ge 0}^2: i=1, \ldots, \delta_j; j=1, \ldots, n \right\},
\end{equation}
{where $\delta_j \ge 1$ is the number of operational modes of job $j$ and $\delta = \max_{1\le j
\le n} \delta_j$. For the total weighted completion time objective introduced below, each job $j$
additionally carries a positive integer weight $w_j$.} Such an instance calls for decisions $(\pi,
\mu)$, where $\pi$ is an $n$-permutation over $[n]\equiv\{1, \ldots, n\}$ and $\mu: [n]\rightarrow
[\delta]$ with $\mu(j)\le\delta_j$. Define
\begin{align}
     & C_k(\pi, \mu) = \sum_{j=1}^{k}p_{\mu(\pi(j)),\pi(j)} +q_{\mu(\pi(k)),\pi(k)}, \ k=1, \ldots, n; \label{eqn:Ck} \\ & C_{\min}(\pi,\mu) =\min_{1\le k\le n} C_k(\pi, \mu); \label{eqn:Cmin} \\ & C_{\max}(\pi,\mu) =\max_{1\le k\le n} C_k(\pi, \mu); \label{eqn:Cmax} \\
    \intertext{and let}
     & D_{\max}(\pi, \mu) =C_{\max}(\pi,\mu)-C_{\min}(\pi,\mu). \label{eqn:Dmax}
\end{align}

In terms of job sequencing terminology, we have $n$ jobs to be processed by a single machine. Each
job $j$ has $\delta_j$ operational modes and consists of two specific operations under each mode,
with the first operation ({the \textit{processing} operation) exclusively occupying the machine. We
refer to the second operation as the \textit{tail} operation}, representing, e.g., cooling, drying,
or waiting for a chemical reaction in practice. The tail operations do not occupy the machine. The
completion time of a job is defined as the sum of its start time, its processing time, and its tail
time. We need to choose a job processing sequence and an operational mode for each job, so that the
length of a certain time window is minimized. {Note that the definition of $C_k(\pi,\mu)$ implies
that the machine starts at time zero and processes the jobs consecutively without intermediate idle
time. This no-idle assumption is immaterial for the regular objectives $C_{\max}$ and $\sum w_j
C_j$ to be minimized, upon which inserted idle time can never improve, but it is essential for
maximizing $C_{\min}$ and for minimizing $D_{\max}$: without it, $C_{\min}$ could be increased
without bound by delaying all jobs, and the complexity of minimizing $D_{\max}$ changes
\citep{Hoogeveen1990}. The variant in which deliberate idle time is permitted is indicated by
``\emph{idle}'' in \Cref{tbl:main-results,tbl:complexity}.}

Following the three-field notation introduced in \Cref{sec:literature-review}, we denote our
problem as $1|\Mode (p_{ij} , q_{ij})|\gamma$, where the machine environment 1 indicates a single
machine, Mode $(p_{ij} , q_{ij} )$ indicates that job $j$ may be processed under any of its
$\delta_j$ modes, each specifying a processing time $p_{ij}$ and a tail time $q_{ij}$, and $\gamma$
is an objective function to be maximized or minimized. We summarize our main results and the
complexity landscape of our model in \Cref{tbl:main-results} and \Cref{tbl:complexity},
respectively, {where the objective $C_{\min}$ is to be maximized and all other objectives are to be
minimized.}

\begin{table}[ht]
    \centering
    \vspace{-0.5cm}
    \caption{Summary of Main Results}\label{tbl:main-results}
    \medskip
    \begin{tabularx}{\textwidth}{llX}
        \toprule
        Problem Descriptor                                          & Complexity           & Reference                                          \\
        \midrule
        $1 | r_j, q_j | C_{\max} $                                  & NP-hard              & \cite{lenstra1977complexity}                       \\
        \addlinespace
        $1 | q_j | D_{\max} $                                       & Poly solvable        & \cite{Hoogeveen1990}                               \\
        \addlinespace
        $1 | q_j, \textup{idle} | D_{\max} $                        & Strongly NP-hard     & \cite{Hoogeveen1990}
        \\
        \addlinespace
        $1 | q_j | C_{\max} $                                       & Poly solvable        & \cite{woeginger1994heuristics}
        \\
        \addlinespace
        $1 | q_j | \sum w_j C_j$                                    & Poly solvable        & \Cref{lem:d1wjcj}                                  \\
        \addlinespace
        $1 | q_j | C_{\min}$                                        & Poly solvable        & \Cref{lem:d1Cmin}                                  \\
        \addlinespace
        $1 | \Mode(p_{ij},q_{ij}) | \sum C_j$                       & Poly solvable        & \cite{CHEN199769}                                  \\
        \addlinespace
        $1 | \Mode(p_{ij},q_{ij}),\ \bar{\pi} | C_{\max}$           & Poly solvable        & \Cref{lem:fix-pi-Cmax}                             \\
        \addlinespace
        $1 | \Mode(p_{ij},q_{ij}),\ \bar{\pi} | C_{\min}$           & Poly solvable        & \Cref{lem:fix-pi-maxCmin}                          \\
        \addlinespace
        $1 | \Mode(p_{ij},q_{ij}),\ \bar{\pi} | \sum w_j C_j$       & Poly solvable        & \Cref{lem:fix-pi-sum-wjcj}                         \\
        \addlinespace
        $1 | \Mode(p_{ij}, q_{ij}), \bar{\pi}, \delta=2 | D_{\max}$ \qquad \qquad & Weakly NP-hard       & \Cref{them:fix-pi-f,thm:fix-pi-f-var-d}            \\
        \addlinespace
        $1 | \Mode(p_{ij}, q_{ij}), \bar{\pi} | D_{\max}$           & Weakly NP-hard and \qquad \qquad   & \Cref{thm:fix-pi-f-var-d}            \\
                                                                    & pseudo-poly solvable &                                                    \\
        \addlinespace
        $1 | \Mode(p_{ij}, q_{ij}), \delta=2 | C_{\max}$            & NP-hard              & \Cref{thm:d2Cmax}                                  \\
        \addlinespace
        $1 | \Mode(p_{ij}, q_{ij}), \delta=2 | D_{\max}$            & NP-hard              & \Cref{cor:d2f}                                     \\
        \addlinespace
        $1 | \Mode(p_{ij}, q_{ij}), \delta=2 | \sum w_j C_j$        & NP-hard              & \cite{cao2006}                                     \\
        \addlinespace
        $1 | \Mode(p_{ij}, q_{ij}), \delta=2 | C_{\min}$            & NP-hard              & \Cref{thm:d2Cmin}                                  \\
        \addlinespace
        $1 | \Mode(p_{ij}, q_{ij}) | C_{\max}$                      & Strongly NP-hard     & \Cref{thm:Cmax-strong}                             \\
        \addlinespace
        $1 | \Mode(p_{ij}, q_{ij}) | D_{\max}$                      & Strongly NP-hard     & \Cref{cor:f-strong}                                \\
        \addlinespace
        $1 | \Mode(p_{ij}, q_{ij}) | \sum w_j C_j$                  & Strongly NP-hard     & \Cref{thm:wjcj-strong}                             \\
        \addlinespace
        $1 | \Mode(p_{ij}, q_{ij}) | C_{\min}$                      & Strongly NP-hard     & \Cref{thm:Cmin-strong}                             \\
        \bottomrule
    \end{tabularx}
\end{table}

\begin{table}[ht]
    \centering
    \begin{threeparttable}
        \vspace{-0.5cm}
        \caption{Summary of Complexity}\label{tbl:complexity}
        \medskip
        \begin{tabular}{cccccc}
            \toprule
            $\delta$     & $\sum C_j$                & $C_{\max}/\bar{\pi}$                   & $C_{\min}/\bar{\pi}$          & $\sum w_j C_j/\bar{\pi}$                & $D_{\max}/\bar{\pi}$/$idle$                                         \\
            \midrule
            $1$          & $\downarrow$              & $\mathcal{P}$/$\downarrow$ & $\mathcal{P}$/$\downarrow$    & $\mathcal{P}$/$\downarrow$             & $\mathcal{P}$/$\mathcal{P}$/$\mathcal{H}^+$ \\
            \addlinespace
            $2$          & $\downarrow$              & $\mathcal{H}$/$\downarrow$             & $\mathcal{H}$/$\downarrow$    & $\mathcal{H}$/$\downarrow$ & $\mathcal{H}$/$\mathcal{H^-}$/$\uparrow$                            \\
            \addlinespace
            unrestricted & $\mathcal{P}$ & $\mathcal{H}^+$/$\mathcal{P}$          & $\mathcal{H}^+$/$\mathcal{P}$ & $\mathcal{H}^+$/$\mathcal{P}$          & $\mathcal{H}^+$/$\mathcal{H^-}$/$\uparrow$                          \\
            \bottomrule
        \end{tabular}

        \begin{tablenotes}
            \small
            \item[a] $\mathcal{P}$ denotes polynomial-time solvability.
            \item[b] $\mathcal{H^-}$ denotes NP-hardness in the ordinary sense.
            \item[c] $\mathcal{H}$ denotes at least NP-hardness in the ordinary sense.
            \item[d] $\mathcal{H}^+$ denotes NP-hardness in the strong sense.
            \item[e] $\bar{\pi}$ denotes the same problem with fixed job sequences.
            \item[f] The row for constant $\delta$ is omitted due to apparent implications.
        \end{tablenotes}
    \end{threeparttable}
\end{table}
% Jelmer: The meaning of the blue entries is never explained. If blue is meant to mark results taken from, or implied by, the literature rather than proved in this paper.

\subsection{Preliminary lemmas for single-mode case}\label{sec:single-mode}

In this subsection, we focus on the special case of $\delta =1$. In other words, each job has only one operational mode. Therefore, specification for job $j$ in \eqref{eqn:general-instance} is simplified as $(p_j, q_j)$.

\begin{lemma}[\cite{woeginger1994heuristics}]\label{lem:d1Cmax}
    The problem $1 | q_{j} | C_{\max}$ can be solved to optimality in $O(n\log n)$ time by sequencing all jobs in non-increasing order of $q_{j}$. $\Box$
\end{lemma}

\begin{lemma}\label{lem:d1wjcj}
    The problem $1 | q_j | \sum w_j C_j$ can be solved to optimality in $O(n\log n)$ time by sequencing the jobs in non-increasing order of the ratio $w_j/p_j$.
\end{lemma}

\begin{proof}[Proof sketch]
    First, note that moving all jobs $j$ with $p_j=0$ to the beginning in arbitrary order from any schedule will not increase the objective function value. For the schedule of all remaining jobs with a positive processing time each, if job $j$ immediately precedes job $k$ with $w_j/p_j < w_k/p_k$, then swapping these two jobs in the schedule will decrease the objective function value.
\end{proof}

\begin{lemma}\label{lem:d1Cmin}
    The problem $1 | q_j | C_{\min}$ can be solved to optimality by sequencing all jobs in non-increasing order of $p_j+q_j$.
\end{lemma}

\begin{proof}[Proof sketch]
    The objective of this problem is to maximize $C_{\min}(\pi)$ defined in \eqref{eqn:Cmin} over all $n$-permutations $\pi$ ({the mode assignment $\mu$ is trivial since $\delta=1$}). {Recall that our model does not permit machine idle time; otherwise $C_{\min}$ could be increased without bound by delaying all jobs.} It is easy to verify that, if job $j$ immediately {precedes} job $k$ in $\pi$ with $p_j+q_j < p_k+q_k$, which implies that $C_j < C_k$, then after swapping these two jobs with all other jobs fixed, the completion time of either job will be at least $C_j$, while the completion time of any other job remains the same. Therefore, starting from any schedule $\pi$, repeated swapping of any two adjacent jobs if necessary will result in a schedule $\pi'$ with $C_{\min}(\pi')\ge C_{\min}(\pi)$ and $p_{\pi'(j)}+q_{\pi'(j)}$ is non-increasing in $j$.
\end{proof}

\subsection{Polynomial solvability under fixed job sequence}\label{sec:fixed-sequence}

This subsection examines the case in which the job permutation is imposed exogenously, for instance, by upstream process constraints in a BPLC printing line where cells must be deposited in a prescribed order. Hence, only the mode assignment remains to be decided for each job. This setting is of independent interest because in many maturation-based production environments the sequence in which items enter the machine is fixed by considerations outside the scheduler's control, while the choice of operational mode remains flexible. We show below that, once the sequence is fixed, which is indicated by adding $\bar{\pi}$ into the middle field in the scheduling problem descriptor, three of our four objectives admit efficient exact algorithms, each exploiting a different structural property of the corresponding objective. These results delineate a practically relevant tractable regime and, as we show in \Cref{sec:weakly-np-hard}, also serve as essential building blocks for several of the hardness reductions that follow.

{The three proofs below share the same notation. Since the sequence is fixed, we relabel the jobs
so that the imposed order is $\bar{\pi}=(1,\dots,n)$; a solution is then specified solely by a mode
assignment $\mu$ with $\mu(j)\in\{1,\dots,\delta_j\}$. For such an assignment we abbreviate
$p_j=p_{\mu(j),j}$ and $q_j=q_{\mu(j),j}$, write $S_j=\sum_{k=1}^{j} p_k$ for the cumulative
processing time through job $j$ (with $S_0=0$), and denote the completion time of job $j$ by
$C_j=S_{j-1}+p_j+q_j$.}

\begin{lemma}\label{lem:fix-pi-Cmax}
    The problem $1|\Mode(p_{ij},q_{ij}),\bar \pi|C_{\max}$ is polynomially solvable.
\end{lemma}

\begin{proof}
    {The makespan is $C_{\max}=\max_j C_j$, which we minimize by binary search over a polynomial-time feasibility test.} Let $p_{\max} = \max_{i,j} p_{ij}$ and $q_{\max} = \max_{i,j} q_{ij}$. For a given threshold $T$, we ask whether there exists a mode assignment {$\mu$} such that $C_j \le T$ for every job $j$.

    The decision test is greedy. Set $S_0 = 0$. For $j = 1, \dots, n$, define the set of eligible modes
    \begin{align*}
        \mathcal{F}_j(T) = \bigl\{ {i\in\{1,\dots,\delta_j\}} \mid S_{j-1} + p_{ij} + q_{ij} \le T \bigr\}.
    \end{align*}
    If {$\mathcal{F}_j(T) = \varnothing$}, declare infeasible. Otherwise, select a mode $i_j \in \mathcal{F}_j(T)$ with minimal processing time and {set $\mu(j) = i_j$, so that $p_j = p_{\mu(j),j}$; update $S_j = S_{j-1} + p_j$}. If all jobs are processed successfully, declare feasible.

    %We now establish correctness. If the greedy test returns feasible, the constructed assignment satisfies the bound by definition. Conversely, suppose a feasible assignment $i_1^*, \dots, i_n^*$ exists, and let $S_j^*$ be the cumulative processing times under this assignment. We prove by induction that the greedy cumulative times $S_j^g$ satisfy $S_j^g \le S_j^*$ for all $j$. The base case $S_0^g = S_0^* = 0$ is immediate. Assume $S_{j-1}^g \le S_{j-1}^*$. Since $i_j^*$ is feasible, $S_{j-1}^* + p_{i_j^* j} + q_{i_j^* j} \le T$. Using the induction hypothesis, $S_{j-1}^g + p_{i_j^* j} + q_{i_j^* j} \le T$, so $i_j^* \in \mathcal{F}_j(T)$. The greedy rule chooses a mode with the smallest processing time in $\mathcal{F}_j(T)$, hence $p_{i_j^g j} \le p_{i_j^* j}$. Consequently,
    %\(S_j^g = S_{j-1}^g + p_{i_j^g j} \le S_{j-1}^* + p_{i_j^* j} = S_j^*.\)
    %Thus the greedy test never fails whenever a feasible assignment exists, proving its correctness.

    For the given optimization problem, define $\tau_{\min} = \max_j \min_i (p_{ij} + q_{ij})$, and $\tau_{\max} = \sum_j \max_i p_{ij} + q_{\max}$. The optimal makespan $C_{\max}^*$ lies in $[\tau_{\min}, \tau_{\max}]$, and $\tau_{\max} - \tau_{\min} \le n p_{\max} + q_{\max}$. Hence, a binary search over this interval requires $O(\log(n p_{\max} + q_{\max}))$ iterations. Each feasibility test runs in ${O(\sum_j \delta_j)}$ time. The smallest $T$ for which the greedy test returns feasible is exactly $C_{\max}^*$. Therefore, the optimal makespan is found in ${O\bigl((\sum_j \delta_j) \log(n p_{\max} + q_{\max})\bigr)}$ time, which is polynomial in the input size.
\end{proof}

\begin{lemma}\label{lem:fix-pi-maxCmin}
    The problem $1 | \Mode(p_{ij}, q_{ij}), \bar{\pi} | C_{\min}$ is polynomially solvable.
\end{lemma}

\begin{proof}
    {The objective is $C_{\min}=\min_j C_j$.} We first solve the decision problem: given $L$, does there exist a mode assignment {$\mu$} with $C_{\min}\ge L$? The greedy test initializes $S_0=0$ and, for ${j}=1,\dots,n$, defines
    {\begin{align*}
        \mathcal{F}_j(L)=\{i\in\{1,\dots,\delta_j\}\mid S_{j-1}+p_{ij}+q_{ij}\ge L\}.
    \end{align*}}
    {If $\mathcal{F}_j(L)=\varnothing$, stop and reject; otherwise choose $i_j\in\mathcal{F}_j(L)$ of maximal processing time, set $\mu(j)=i_j$ so that $p_j=p_{\mu(j),j}$, and update $S_j=S_{j-1}+p_j$. Accept if all jobs are processed.}

    %Correctness follows by an induction analogous to that of Lemma~\ref{lem:fix-pi-Cmax} (with inequalities reversed): if a feasible assignment $i^*$ exists with cumulative times $S_k^*$, then the greedy times $S_k^g$ satisfy $S_k^g\ge S_k^*$ for all $k$. Hence $i_k^*\in\mathcal{F}_k$ whenever $S_{k-1}^g\ge S_{k-1}^*$; the greedy choice gives $p_{i_k^g k}\ge p_{i_k^* k}$, so $S_k^g\ge S_k^*$. Thus the test accepts exactly when a feasible assignment exists.

    For optimization, note that $C_{\min}^*\in[0,U]$ with $U=\sum_j\max_i p_{ij}+\max_{i,j}q_{ij}$. Binary search over this interval requires $O(\log U)$ iterations; each test runs in $O(\sum_j\delta_j)$. {Since $\log U$ is bounded by a polynomial in the input size, the total time $O((\sum_j\delta_j)\log U)$ is polynomial.}
\end{proof}

\begin{lemma}\label{lem:fix-pi-sum-wjcj}
    The problem $1 |\Mode(p_{ij}, q_{ij}), \bar{\pi} | \sum w_j C_j$ is polynomially solvable.
\end{lemma}

\begin{proof}
    {The completion time of job $j$ is}
    \begin{equation*}
        C_j = \sum_{k=1}^{j} p_k + q_j .
    \end{equation*}
    The total weighted completion time can be rewritten as
    \begin{equation*}
        \sum_{j=1}^n w_j C_j
        = \sum_{j=1}^n w_j \Bigl( \sum_{k=1}^{j} p_k + q_j \Bigr) = \sum_{j=1}^n p_j \Bigl( \sum_{k=j}^n w_k \Bigr) + \sum_{j=1}^n w_j q_j .
    \end{equation*}
    Set $W_j = \sum_{k=j}^n w_k$ for $j=1,\dots,n$; these values are independent of the mode choices. If job $j$ is assigned mode $i$, its contribution to the objective becomes $p_{ij}W_j + w_j q_{ij}$. The total cost is therefore the sum of $n$ independent job-wise contributions. Hence the problem decomposes into $n$ independent sub-problems: for each job $j$, {set $\mu(j)$ to} a mode $i \in \{1,\dots,\delta_j\}$ that minimizes $p_{ij}W_j + w_j q_{ij}$.  {The assignment $\mu$ obtained in this way is optimal}, and the minimal total weighted completion time equals
    \begin{equation*}
        \sum_{j=1}^n \min_{1\le i\le \delta_j} \bigl( p_{ij} W_j + w_j q_{ij} \bigr).
    \end{equation*}
    Computing $W_1,\dots,W_n$ requires $O(n)$ time, and for each job we evaluate at most $\delta_j$ candidate modes. The overall algorithm runs in $O\bigl(n + \sum_{j=1}^n \delta_j\bigr)$ time, which is clearly polynomial in the input size. Thus the problem can be solved in polynomial time.
\end{proof}

\section{Weakly NP-hard problems}\label{sec:weakly-np-hard}

\Cref{sec:fixed-sequence} shows that fixing the job sequence renders mode selection tractable for
objectives $C_{\max}$, $C_{\min}$ and $\sum w_j C_j$. We now ask two questions that are left open:
does the same tractability extend to the range objective $D_{\max}$ under a fixed sequence, and
what happens to all four objectives once the sequence itself becomes a decision variable while the
number of modes per job is restricted to the smallest nontrivial case, $\delta = 2$? {This section
answers both questions in the negative. Every hardness result in this section is obtained via a
reduction from the \textsc{Partition} problem, which is weakly NP-complete
\citep{garey1979computers}.} \Cref{sec:delta2+Dmax,sec:fix-pi+Dmax} hold the sequence fixed and
isolate $D_{\max}$ as the source of difficulty; \Cref{sec:delta2+Cmax,sec:delta2+Cmin} release the
sequence but cap the number of modes at two, showing that {the objectives $C_{\max}$ and
$C_{\min}$, both solvable by simple sorting rules in the single-mode case, already become hard once
a job can choose between just two alternatives.}

{Throughout this section, an instance of \textsc{Partition} consists of a set of positive integers $s_1,\dots,s_n$ with $\sum_{j=1}^n s_j = 2B$ and asks whether there exists a subset $S\subseteq\{1,\dots,n\}$ such that $\sum_{j\in S}s_j = B$. Without loss of generality we assume $s_j < B$ for all $j$: an instance containing an element $s_j > B$ is trivially a no-instance, and one containing an element $s_j = B$ is trivially a yes-instance.} % Jelmer: I think we can just remove this paragraph. The reader should know what Partition is.

\subsection{Problem $1 | \Mode(p_{ij}, q_{ij}), \bar{\pi}, \delta=2 | D_{\max}$}\label{sec:delta2+Dmax}

We begin with the most restrictive setting in this paper: the sequence is fixed at $\bar{\pi}$, as
in \Cref{sec:fixed-sequence}, and each job is limited to $\delta = 2$ modes. \Cref{lem:fix-pi-Cmax}
shows that under these conditions $C_{\max}$ can still be optimized in polynomial time via a greedy
feasibility test that exploits the fact that a threshold on $C_{\max}$ imposes only an upper bound
on each job's completion time. The range objective $D_{\max}$, however, imposes both an upper
\emph{and} a lower bound on every completion time simultaneously, and {this two-sided requirement
is enough to break the monotonicity argument underlying the greedy algorithm of
\Cref{lem:fix-pi-Cmax}. As the following theorem shows, tractability is not preserved under this
change of objective.}

\begin{theorem}\label{them:fix-pi-f}
    Problem $1 | \Mode(p_{ij},q_{ij}), \bar\pi, \delta=2 | D_{\max}$ is NP-hard.
\end{theorem}

\begin{proof}
    Membership in NP is straightforward: given a mode assignment $\mu$, the completion times{, and hence $C_{\max}$, $C_{\min}$, and $D_{\max}=C_{\max}-C_{\min}$,} can be computed in polynomial time. We prove NP-hardness by a reduction from \textsc{Partition}. {Given an instance $s_1,\dots,s_n$ of \textsc{Partition}, we build a scheduling instance whose fixed job order is}
    \begin{align*}
        A_L,\ A_H,\ J_1,\dots,J_n,\ F_L,\ F_H .
    \end{align*}
    The two anchor jobs have a single mode each:
    \begin{align*}
        A_L:\ (p,q)=(0,2B),\qquad A_H:\ (p,q)=(0,4B).
    \end{align*}
    For each $s_j$, create an ordinary job $J_j$ with two modes:
    \begin{align*}
        \text{mode }1:\ (p,q)=(s_j,2B),\qquad
        \text{mode }2:\ (p,q)=(0,2B).
    \end{align*}
    Finally, add two single-mode testing jobs:
    \begin{align*}
        F_L:\ (p,q)=(0,B),\qquad F_H:\ (p,q)=(0,3B),
    \end{align*}
    and set the threshold $K = 2B$.

    For a given mode assignment $\mu$, let $S = \{j \mid J_j \text{ is assigned mode }1\}$ and $X = \sum_{j\in S} s_j$. The total processing time of the ordinary jobs equals $X$, because mode~$1$ contributes $s_j$ and mode~$2$ contributes $0$.

    The completion times of the anchor jobs are $C(A_L)=2B$ and $C(A_H)=4B$. Consequently every schedule satisfies $C_{\min}\le 2B$, $C_{\max}\ge 4B$, and therefore
    \begin{align*}
        D_{\max}=C_{\max}-C_{\min}\ge 2B .
    \end{align*}
    Thus a schedule meets $D_{\max}\le K=2B$ if and only if all completion times lie in the interval $[2B,4B]$.

    Now examine an ordinary job $J_j$. Because $J_j$ appears after $A_L$ and $A_H$, its completion time equals the cumulative processing time of all jobs up to itself (including $J_j$) plus its tail $2B$. The cumulative processing time at $J_j$ is {a partial sum $P_j$ of the selected sizes, with $0\le P_j\le 2B$}, hence
    \begin{align*}
        C(J_j) = P_j + 2B \in [2B,4B].
    \end{align*}
    Thus ordinary jobs never violate the required interval.

    It remains to check the two testing jobs. They have zero processing time and start after all ordinary jobs, so their completion times are
    \begin{align*}
        C(F_L) = X + B,\qquad C(F_H) = X + 3B .
    \end{align*}
    For both to lie in $[2B,4B]$ we need
    \begin{align*}
        X + B \ge 2B \quad\text{and}\quad X + 3B \le 4B,
    \end{align*}
    which is equivalent to $X = B$.

    We now show the equivalence. If the \textsc{Partition} instance has a solution $S$ with $\sum_{j\in S}s_j = B$, assign mode~$1$ to $J_j$ for $j\in S$ and mode~$2$ to the remaining ordinary jobs. Then $X = B$, and all completion times fall into $[2B,4B]$ with $A_L$ and $A_H$ attaining the endpoints. Hence $C_{\min}=2B$, $C_{\max}=4B$, and $D_{\max}=2B=K$.

    Conversely, suppose the constructed instance admits a mode assignment with $D_{\max}\le K$. Since the anchor jobs force $D_{\max}\ge 2B$, we must have $D_{\max}=2B$ and every completion time in $[2B,4B]$. In particular, $C(F_L)=X+B\ge 2B$ and $C(F_H)=X+3B\le 4B$, which gives $X = B$. The set of jobs assigned mode~$1$ then constitutes a subset $S$ with $\sum_{j\in S}s_j = B$, solving
    \textsc{Partition}.

    {The construction uses $n+4$ jobs with at most two modes per job and can be carried out in time polynomial in the size of the \textsc{Partition} instance.} The decision problem is therefore NP-complete, and the corresponding optimization problem is NP-hard.
\end{proof}

\begin{remark}
    {In terms of the feasibility test of \Cref{lem:fix-pi-Cmax}, a threshold on $C_{\max}$ imposes only upper bounds $S_{j-1}+p_{ij}+q_{ij}\le T$, so choosing a feasible mode with smallest processing time preserves future feasibility, whereas a threshold $K$ on $D_{\max}$ imposes the two-sided window constraints $L\le S_{j-1}+p_{ij}+q_{ij} \le L+K$, whose lower bounds destroy this monotonicity. The reduction above exploits exactly the two-sided window requirement to encode the equality $\sum_{j\in S}s_j=B$.}
\end{remark}

\subsection{Problem $1 | \Mode(p_{ij}, q_{ij}), \bar{\pi} | D_{\max}$}\label{sec:fix-pi+Dmax}

\Cref{them:fix-pi-f} shows that, even with only two modes per job, minimizing $D_{\max}$ under a
fixed sequence is NP-hard. A natural question is whether relaxing $\delta = 2$ to an unrestricted
number of modes pushes this hardness from weak to strong, which would mirror the pattern we later
observe for the free-sequence problems of \Cref{sec:strongly-np-hard}, where lifting the mode-count
restriction upgrades {the weak hardness of \Cref{thm:d2Cmax} for $C_{\max}$ into the strong
hardness of \Cref{thm:Cmax-strong}}. Somewhat surprisingly, the answer is negative, as the
following theorem demonstrates.

\begin{theorem}\label{thm:fix-pi-f-var-d}
    Problem $1 | \Mode(p_{ij},q_{ij}),\bar\pi | D_{\max}$ admits a pseudo-polynomial-time algorithm, even when the number of modes is part of the input.
\end{theorem}

\begin{proof}
    Since the permutation is fixed, we relabel the jobs so that the processing order is $\bar\pi=(1,2,\dots,n)$. For a mode assignment $\mu$, define the cumulative processing time after job $k$ as
    \begin{align*}
        S_k(\mu)=\sum_{j=1}^k p_{\mu(j),j},
    \end{align*}
    and the {completion time of job $k$} as
    \begin{align*}
        U_k(\mu)=S_k(\mu)+q_{\mu(k),k}.
    \end{align*}
    The objective can then be written as
    \begin{align*}
        D_{\max}(\bar\pi,\mu) = \max_{1\le k\le n} U_k(\mu) - \min_{1\le k\le n} U_k(\mu).
    \end{align*}

    Consider the decision version: given an integer $K$, determine whether there exists a mode assignment $\mu$ with $D_{\max}(\bar\pi,\mu)\le K$. This holds exactly when all $U_k(\mu)$ lie in some interval $[\ell,\ell+K]$. Let
    \begin{align*}
        A = \sum_{j=1}^n \max_{1\le i\le \delta_j} p_{ij}, \qquad B = \max_{i,j} q_{ij}.
    \end{align*}
    Every cumulative processing time belongs to $\{0,1,\dots,A\}$, hence every possible value of $\min_k U_k(\mu)$ lies in $\{0,1,\dots,A+B\}$. It therefore suffices to enumerate all candidate lower endpoints $\ell \in \{0,1,\dots,A+B\}$ and test, for each $\ell$, whether a feasible mode assignment exists with $U_k(\mu)\in[\ell,\ell+K]$ for all $k$.

    For a fixed $\ell$, we test feasibility by dynamic programming. Let $R_k(\ell) \subseteq \{0,1,\dots,A\}$ be the set of cumulative processing times attainable after jobs $1,\dots,k$ while keeping $U_1,\dots,U_k$ inside $[\ell,\ell+K]$. Initialize $R_0(\ell)=\{0\}$. For $k=1,\dots,n$, we build $R_k(\ell)$ from $R_{k-1}(\ell)$:
    \begin{align*}
        R_k(\ell)= \bigl\{s + p_{ik} \mid s\in R_{k-1}(\ell), 1\le i\le \delta_k,
        \ell \le s + p_{ik} + q_{ik} \le \ell+K \bigr\}.
    \end{align*}
    If $R_n(\ell)\neq\varnothing$, the window $[\ell,\ell+K]$ is feasible; the original decision problem is a yes-instance precisely when this happens for some $\ell$.

    Correctness is straightforward. If the dynamic program produces a nonempty $R_n(\ell)$, the sequence of modes that generated the surviving state yields $U_k(\mu)\in[\ell,\ell+K]$ for all $k$, so $D_{\max}(\bar\pi,\mu)\le K$. Conversely, if some mode assignment $\mu$ satisfies $D_{\max}(\bar\pi,\mu)\le K$, set $\ell = \min_k U_k(\mu)$. Then $\ell \le U_k(\mu)\le \ell+K$ for every $k$, and an induction on $k$ shows $S_k(\mu)\in R_k(\ell)$, hence $R_n(\ell)\neq\varnothing$.

    For one value of $\ell$, the dynamic program examines at most $A+1$ states per job, and for each state considers at most $\delta_k$ transitions. Thus the time per $\ell$ is $O\bigl(A\sum_{k=1}^n \delta_k\bigr)$. With $A+B+1$ choices of $\ell$, the total running time becomes
    \begin{align*}
        O\Bigl((A+B+1)A\sum_{k=1}^n \delta_k\Bigr),
    \end{align*}
    which is pseudo-polynomial in the numerical input values.
\end{proof}

\subsection{Problem $1 | \Mode(p_{ij}, q_{ij}), \delta=2 | C_{\max}$}\label{sec:delta2+Cmax}

We now turn from the fixed-sequence setting of \Cref{sec:delta2+Dmax,sec:fix-pi+Dmax} to the
general problem in which both the permutation and the mode assignment must be determined jointly,
and ask how few modes are needed before this joint decision becomes hard. The answer is as small as
possible: already at $\delta = 2$, minimizing $C_{\max}$ (which is trivially solvable in $O(n \log
n)$ time when $\delta = 1$ by \Cref{lem:d1Cmax}) becomes NP-hard (\Cref{thm:d2Cmax}). The proof
exploits the fact that the scheduler must now also decide where in the sequence the mode-1 and
mode-2 jobs belong. We show {that some optimal schedule} groups jobs by non-increasing tail time,
allowing a Partition instance to be embedded in the resulting two-block structure. We then show,
via a direct reduction from this $C_{\max}$ result, that $D_{\max}$ is likewise NP-hard at $\delta
= 2$ (\Cref{cor:d2f}), by appending a single auxiliary job whose large tail forces it to open the
schedule in any optimum, collapsing the $D_{\max}$ objective on the augmented instance to the
$C_{\max}$ objective on the original one.

\begin{theorem}\label{thm:d2Cmax}
    Problem $1 | \Mode(p_{ij}, q_{ij}), \delta=2 | C_{\max}$ is at least weakly NP-hard.
\end{theorem}

\begin{proof}
    We reduce from \textsc{Partition}. Given an instance, construct a scheduling instance with $n$ jobs. For each $j$, define two modes:
    \begin{align*}
        \text{Mode 1: } (p_{1j},q_{1j})=(s_j,2B),\qquad
        \text{Mode 2: } (p_{2j},q_{2j})=(2s_j,0).
    \end{align*}
    Let $x=\sum_{j\in M_1}s_j$, where $M_1$ is the set of jobs assigned to mode 1. Then $\sum_{j\in M_2}2s_j = 2(2B-x)=4B-2x$.

    By the standard interchange argument (\Cref{lem:d1Cmax}), there exists an optimal schedule in which all mode-1 jobs precede all mode-2 jobs, with arbitrary order within each group. Under such a schedule, the last mode-1 job completes at $C_1=x+2B$, while the last mode-2 job completes at $C_2=x+(4B-2x)=4B-x$. Hence the makespan is
    \begin{align*}
        C_{\max}=\max\{x+2B, 4B-x\}.
    \end{align*}
    Since $x$ is a subset sum of the $s_j$'s, observe that
    \begin{align*}
        C_{\max} = \begin{cases} 4B-x \ge 3B, & x\le B, \\ x+2B \ge 3B, & x\ge B,
                   \end{cases}
    \end{align*}
    with equality to $3B$ if and only if $x=B$.

    If \textsc{Partition} has a solution $S$, assigning mode 1 to $S$ yields $x=B$ and hence $C_{\max}=3B$. Conversely, if the scheduling instance admits a schedule with makespan {at most} $3B$, then $C_{\max}=\max\{x+2B,4B-x\}=3B$, forcing $x=B$. The jobs in $M_1$ then form a valid subset summing to $B$.

    The reduction is polynomial. Since \textsc{Partition} is NP-complete, the theorem follows.
\end{proof}

\begin{corollary}\label{cor:d2f}
    Problem $1 | \Mode(p_{ij}, q_{ij}), \delta=2 | D_{\max}$ is at least NP‑hard in the ordinary sense.
\end{corollary}

\begin{proof}
    We reduce from the NP-hard problem $1|\Mode(p_{ij},q_{ij}),\delta=2|C_{\max}$ (\Cref{thm:d2Cmax}). Let $\mathcal{I}$ be an instance with $n$ jobs, each having at most two modes with parameters $(p_{ij},q_{ij})$. Define
    \begin{align*}
        A=\sum_{j=1}^n \max_i p_{ij},\qquad B=\max_{i,j} q_{ij},\qquad M=2(A+B)+1.
    \end{align*}
    Construct $\mathcal{I}'$ by retaining all original jobs and adding a new single-mode job $J_0$ with $(p_0,q_0)=(M,0)$. {The construction is clearly polynomial.}

    Consider any schedule of $\mathcal{I}'$. If $J_0$ is not first, then the first job completes by time $A+B$, while $J_0$ completes no earlier than $M$; hence $D_{\max}\ge M-(A+B)=A+B+1$. If $J_0$ is first, then it completes at $M$, and
    \begin{align*}
        D_{\max}=C_{\max}^{\mathcal{I}},
    \end{align*}
    where $C_{\max}^{\mathcal{I}}$ is the makespan of the induced schedule on the original jobs.

    Since every schedule of $\mathcal{I}$ has makespan at most $A+B$, any optimal schedule of $\mathcal{I}'$ must place $J_0$ first. Therefore $\mathrm{OPT}_{D}(\mathcal{I}')=\mathrm{OPT}_{C}(\mathcal{I})$. It follows that $\mathcal{I}$ admits a schedule with $C_{\max}\le T$ if and only if $\mathcal{I}'$ admits a schedule with $D_{\max}\le T$. The reduction is polynomial, so the claimed NP-hardness follows.
\end{proof}

\subsection{Problem $1 | \Mode(p_{ij}, q_{ij}), \delta=2 | C_{\min}$}\label{sec:delta2+Cmin}

We complete our analysis on $\delta = 2$ with the objective $C_{\min}$, which asks for a schedule maximizing the earliest completion time among all jobs. As in \Cref{sec:delta2+Cmax}, this objective is already NP-hard with two modes per job, even though the single-mode version (\Cref{lem:d1Cmin}) is solvable by a simple sorting rule. The reduction again relies on an exchange argument showing that an optimal schedule orders jobs by non-increasing $p_{ij} + q_{ij}$, then embeds a Partition instance in the resulting two-block schedule, which is structurally similar to the reduction of \Cref{thm:d2Cmax}, but calibrated to $C_{\min}$ rather than $C_{\max}$.

\begin{theorem}\label{thm:d2Cmin}
    The single-machine scheduling problem $1| \Mode(p_{ij}, q_{ij}), \delta=2| C_{\min}$ is at least weakly NP-hard.
\end{theorem}

\begin{proof}
    We reduce from \textsc{Partition}{, recalling the assumption $s_j < B$ for all $j$}. Given an instance, construct a scheduling instance with $n+1$ jobs. For each $j=1,\dots,n$, define two modes:
    \begin{align*}
        \text{Mode 1: } (p_{1j},q_{1j})=(s_j, 3B-s_j),\qquad
        \text{Mode 2: } (p_{2j},q_{2j})=(2s_j, 2B-2s_j).
    \end{align*}
    Add an auxiliary job $J_{n+1}$ with a single mode $(p,q)=(0,0)$. All parameters are nonnegative integers since $s_j<B$. The decision version asks whether there exists a schedule with $\min_j C_j \ge 3B$.

    {For any fixed mode assignment, there exists an optimal sequence that orders the jobs by non-increasing $p_{ij}+q_{ij}$ (\Cref{lem:d1Cmin}).} In our instance, $p_{1j}+q_{1j}$ equals $3B$ for Mode 1 jobs, $p_{2j}+q_{2j}$ equals $2B$ for Mode 2 jobs, and {the sum of processing and tail time equals $0$ for job $J_{n+1}$}. Hence an optimal schedule processes all Mode-1 jobs first, then all Mode-2 jobs, and finally $J_{n+1}$; the order within each group is arbitrary.

    Let $S\subseteq\{1,\dots,n\}$ be the set of jobs assigned Mode 1, and let $X=\sum_{j\in S}s_j$. The remaining jobs are assigned Mode 2. Under the canonical order, the Mode-1 jobs start at time $0$ and have minimum completion time $3B$; the Mode-2 jobs start at time $X$ and have minimum completion time $X+2B$; and $J_{n+1}$ completes at $4B-X$. Thus
    \begin{align*}
        \min_j C_j = \min\{3B, X+2B, 4B-X\}.
    \end{align*}
    Since $X\in[0,2B]$, we have
    \begin{align*}
        \min_j C_j = \begin{cases}
                         X+2B, & X\le B, \\ 4B-X, & X\ge B,
                     \end{cases}
    \end{align*}
    so $\min_j C_j \le 3B$, with equality if and only if $X=B$.

    If \textsc{Partition} has a solution $S$, assigning Mode 1 to $S$ yields $X=B$ and hence $\min_j C_j=3B$. Conversely, if a schedule attains $\min_j C_j\ge 3B$, then by the structural property we may assume the canonical order without loss of optimality; the above formula gives $\min\{X+2B,4B-X\}\ge 3B$, which forces $X=B$. The jobs in $S$ then form a valid \textsc{Partition} solution.

    Thus the decision version is NP-complete, and the optimization problem is at least weakly NP-hard.
\end{proof}

\section{Strongly NP-hard problems}\label{sec:strongly-np-hard}

The results of \Cref{sec:weakly-np-hard} show that once mode selection and sequencing are decided
jointly, hardness sets in already at $\delta = 2${; however, since every reduction there originates
from \textsc{Partition}, those results establish NP-hardness in the weak sense only}. We now show
that when $\delta$ is unrestricted, treated as part of the input rather than a fixed constant,
{every objective in this paper becomes strongly NP-hard}. Each result is obtained via a reduction
from the strongly NP-complete \textsc{3-Partition} problem, using a number of modes per job that
grows with the size of the \textsc{3-Partition} instance.

{Throughout this section, an instance of \textsc{3-Partition} consists of $3n$ positive integers
$s_1,\dots,s_{3n}$ and a positive integer $B$ such that
\begin{equation*}
    \frac{B}{4} < s_j < \frac{B}{2}\quad (1\le j\le 3n),\qquad
    \sum_{j=1}^{3n} s_j = nB ,
\end{equation*}
and asks whether $\{1,\dots,3n\}$ can be partitioned into $n$ disjoint triples $T_1,\dots,T_n$ with $\sum_{j\in T_k}s_j = B$ for each $k$; this problem is NP-complete in the strong sense \citep{garey1979computers}.}

\subsection{Problem $1 | \Mode(p_{ij}, q_{ij}) | C_{\max}$}

We show that when $\delta$ is unrestricted, minimizing $C_{\max}$ becomes strongly NP-hard, in sharp contrast to \Cref{thm:d2Cmax}'s weak hardness at $\delta = 2$.

\begin{theorem}\label{thm:Cmax-strong}
    The problem $1 | \Mode(p_{ij}, q_{ij}) | C_{\max}$ is strongly NP-hard when the parameter $\delta$ is unrestricted.
\end{theorem}

\begin{proof}
    We give a polynomial-time reduction from \textsc{3-Partition}. We construct a scheduling instance with $3n$ jobs, one for each $s_j$, and set the number of modes to $\delta=n$. For job $j$ and mode $k$ ($1\le k\le n$) define
    \begin{equation}\label{eq:pkj}
        p_{kj} = k\cdot s_j ,\qquad q_{kj} = \Bigl(\sum_{\ell=k+1}^{n} \ell\Bigr) B = \frac{(k+1+n)(n-k)}{2}\cdot B .
    \end{equation}
    In particular $q_{nj}=0$ for every $j$, and the tail values satisfy $q_{1j}>q_{2j}>\cdots>q_{nj}=0$. Since $q_{kj}$ does not depend on $j$, we write $q_k = q_{kj}$ for each mode $k$. All numerical values are polynomially bounded in the size of the \textsc{3-Partition} instance, so the reduction runs in polynomial time.

    {Because the tail of a job depends only on its assigned mode and decreases strictly with $k$, \Cref{lem:d1Cmax} shows that, for any fixed mode assignment, resequencing the jobs in non-increasing order of tail does not increase the makespan.} Hence we may restrict attention to schedules that process all jobs of mode~$1$ first (in arbitrary order), then all jobs of mode~$2$, and so on up to mode~$n$. Let $M_k$ be the set of jobs assigned to mode $k$, and define
    \begin{equation*}
        X_k = \sum_{j\in M_k} s_j ,\qquad x_k = \sum_{j\in M_k} p_{kj} = k X_k .
    \end{equation*}
    Clearly $\sum_{k=1}^n X_k = \sum_{j=1}^{3n} s_j = nB$. For a schedule respecting the mode order, the completion time of the last job of mode $k$ is
    \begin{equation*}
        C_k = \sum_{i=1}^{k} x_i + q_k ,
    \end{equation*}
    and the makespan is $C_{\max}= \max_{1\le k\le n} C_k$. We set the target value
    \begin{equation*}
        M^* = \frac{n(n+1)}{2}B .
    \end{equation*}

    Suppose first that the \textsc{3-Partition} instance admits a solution, i.e., the $3n$ numbers can be partitioned into triples $T_1,\dots,T_n$ with sum $B$ each. Assign the three jobs corresponding to $T_k$ to mode $k$. Then $X_k=B$, $x_k=kB$, and
    \begin{align*}
        C_k = \sum_{i=1}^{k} iB + q_k = \frac{k(k+1)}{2}B + \frac{(k+1+n)(n-k)}{2}B = \frac{n(n+1)}{2}B = M^* .
    \end{align*}
    Thus the resulting schedule has makespan exactly $M^*$.

    Conversely, assume there exists a schedule with makespan at most $M^*$. By the structural observation we may reorder the jobs so that the mode order is respected without increasing the makespan; hence there is a schedule of the above form with $C_{\max}\le M^*$. For this schedule we have $C_k\le M^*$ for every $k$, i.e.,
    \begin{equation}\label{eq:bound}
        \sum_{i=1}^{k} x_i \le M^* - q_k = \frac{k(k+1)}{2}B \qquad (1\le k\le n).
    \end{equation}
    Define coefficients
    \begin{equation*}
        c_k =
        \begin{cases}
            \dfrac{1}{k(k+1)}, & 1\le k\le n-1, \\[8pt]
            \dfrac{1}{n},      & k=n.
        \end{cases}
    \end{equation*}
    Multiply the $k$-th inequality in \eqref{eq:bound} by $c_k$ and sum over $k=1,\dots,n$:
    \begin{equation*}
        \sum_{k=1}^{n} c_k \sum_{i=1}^{k} x_i
        \le \sum_{k=1}^{n-1} \frac{1}{k(k+1)}\cdot\frac{k(k+1)}{2}B
        + \frac{1}{n}\cdot\frac{n(n+1)}{2}B = \frac{n-1}{2}B + \frac{n+1}{2}B = nB.
    \end{equation*}
    The left-hand side can be rewritten by exchanging the summations:
    \begin{equation*}
        \sum_{k=1}^{n} c_k \sum_{i=1}^{k} x_i
        = \sum_{i=1}^{n} x_i\Bigl(\sum_{k=i}^{n} c_k\Bigr).
    \end{equation*}
    For $1\le i\le n-1$, $\sum_{k=i}^{n} c_k = \sum_{k=i}^{n-1}\frac{1}{k(k+1)} + \frac{1}{n} = \bigl(\frac{1}{i}-\frac{1}{n}\bigr) + \frac{1}{n} = \frac{1}{i}$, while for $i=n$ the sum equals $\frac{1}{n}$. Hence the left-hand side simplifies to $\sum_{i=1}^{n} \frac{x_i}{i} = \sum_{i=1}^{n} X_i = nB$.

    We have obtained $nB \le nB$, so {all inequalities in \eqref{eq:bound}, each weighted by a strictly positive coefficient $c_k$, must in fact be equalities}. Consequently $\sum_{i=1}^{k} x_i = \frac{k(k+1)}{2}B$ for every $k$, which yields $x_k = kB$ and therefore $X_k = B$ for all $k$.

    Each $X_k$ is a sum of distinct numbers $s_j$ strictly between $B/4$ and $B/2$. A sum equal to $B$ can be achieved only with exactly three such numbers. Thus each $M_k$ consists of precisely three jobs whose $s_j$ values sum to $B$, and the sets $M_1,\dots,M_n$ form a solution to the original \textsc{3-Partition} instance.

    We have shown that the \textsc{3-Partition} instance admits a solution if and only if the constructed scheduling instance possesses a schedule with makespan at most $M^*$. Because all numerical data are polynomially bounded, the reduction is polynomial. As \textsc{3-Partition} is strongly NP-complete, the problem $1|\Mode(p_{ij},q_{ij})|C_{\max}$ is strongly NP-hard.
\end{proof}

Just as \Cref{cor:d2f} lifted \Cref{thm:d2Cmax} from $C_{\max}$ to $D_{\max}$ at $\delta = 2$ by appending a single job with a large tail, the same device lifts \Cref{thm:Cmax-strong} to the unrestricted-mode setting.

\begin{corollary}\label{cor:f-strong}
    The problem $1|\Mode(p_{ij},q_{ij})|D_{\max}$ is strongly NP-hard.
\end{corollary}

\subsection{Problem $1|\Mode(p_{ij},q_{ij})|\sum w_j C_j$}

We next consider the objective function $\sum w_j C_j$\@. \cite{cao2006} already established NP-hardness of this objective at $\delta = 2$ (see \Cref{tbl:main-results}), so unlike the other three objectives no separate $\delta = 2$ result is needed here. Our contribution is to show the hardness is strong once $\delta$ is unrestricted. The reduction is again from \textsc{3-Partition} but is structurally more delicate.

\begin{theorem}\label{thm:wjcj-strong}
    The problem $1|\Mode(p_{ij},q_{ij})|\sum w_j C_j$ is strongly NP-hard.
\end{theorem}

\begin{proof}
    We reduce from \textsc{3-Partition}. We construct a scheduling instance with $3n$ jobs, one for each $s_j$, and set the number of modes to $\delta=n$. For job $j$ and mode $k$ ($1\le k\le n$) define
    \begin{align*}
        w_j = 2s_j,\qquad p_{kj} = 2k s_j,\qquad q_{kj} = M - k s_j - B k(2n+1-k),
    \end{align*}
    where $M$ is a large integer (e.g., $M = 3Bn^2$) that guarantees $q_{kj}\ge 0$ for all $j,k$. All numerical values are polynomially bounded, so the reduction is polynomial.

    For any job assigned to mode $k$ we have $p_{kj}/w_j = k$, independent of $j$. In a single‑machine problem without tails, the weighted shortest processing time (WSPT) rule minimizes $\sum w_jC_j$ by sequencing jobs in non‑decreasing order of $p_j/w_j$ (Lemma \ref{lem:d1wjcj}). Consequently, any optimal schedule must process all jobs that share the same mode as a contiguous block, the blocks themselves appearing in order of increasing $k=1,\dots,n$.  The internal order of jobs inside a block does not
    affect the objective.

    Let $M_k$ be the set of jobs assigned to mode $k$ and define
    \begin{align*}
        X_k = \sum_{j\in M_k} s_j .
    \end{align*}
    Since every job belongs to exactly one mode, $\sum_{k=1}^n X_k = nB$.

    Now we compute the total weighted completion time $F = \sum_j w_j C_j$ for an optimal schedule that respects the block structure described above. The cumulative processing time before block $k$ starts is
    \begin{align*}
        P_{k-1} = 2\sum_{i=1}^{k-1} i X_i .
    \end{align*}
    {Consider the jobs of block $k$. List them in any order as $j_1,\dots,j_m$ and write $v_\ell = w_{j_\ell} = 2s_{j_\ell}$ for $\ell=1,\dots,m$; then $\sum_{\ell=1}^m v_\ell = 2X_k$. The $\ell$-th job of the block has processing time $k v_\ell$ and tail $q_{k,j_\ell}$. Its completion time is therefore
    \begin{align*}
        C_{j_\ell} = P_{k-1} + \sum_{r=1}^{\ell-1} k v_r + k v_\ell + q_{k,j_\ell}.
    \end{align*}
    Multiplying by $v_\ell$ and summing over the block gives three natural components. The start-time contribution is
    \begin{align*}
        \sum_{\ell} v_\ell\Bigl(P_{k-1} + \sum_{r=1}^{\ell-1} k v_r\Bigr)
        & = P_{k-1}\sum_\ell v_\ell + \frac{k}{2}\Bigl({\bigl(\sum_\ell v_\ell\bigr)}^2 - \sum_\ell v_\ell^2\Bigr) \\ & = P_{k-1}W_k + \frac{k}{2}W_k^2 - \frac{k}{2}\sum_\ell v_\ell^2,
    \end{align*}
    where we used $W_k = \sum_\ell v_\ell = 2X_k$. The processing-time component is then simply $\sum_\ell v_\ell (k v_\ell) = k\sum_\ell v_\ell^2$. The tail component expands as
    \begin{align*}
        \sum_{\ell} v_\ell q_{k,j_\ell}
        & = \sum_{j\in M_k} 2s_j\bigl(M - k s_j - B k(2n+1-k)\bigr) \\ &= 2M X_k - 2k\sum_{j\in M_k} s_j^2 - 2B k(2n+1-k) X_k .
    \end{align*}
    Adding these three parts, the quadratic terms in the $s_j$ cancel out because $v_\ell = 2s_{j_\ell}$ implies $\sum_\ell v_\ell^2 = 4\sum_{j\in M_k} s_j^2$, whence
    \begin{align*}
        -\frac{k}{2}\sum_\ell v_\ell^2 + k\sum_\ell v_\ell^2 - 2k\sum_{j\in M_k} s_j^2 = \frac{k}{2}\Bigl(4\sum_{j\in M_k} s_j^2\Bigr) - 2k\sum_{j\in M_k} s_j^2 = 0 .
    \end{align*}}
    Thus the total contribution of block $k$ simplifies to
    \begin{align*}
        \sum_{j\in M_k} w_j C_j = P_{k-1}W_k + \frac{k}{2}W_k^2 + 2M X_k - 2B k(2n+1-k) X_k .
    \end{align*}
    Substituting $W_k = 2X_k$ and $P_{k-1} = 2\sum_{i<k} i X_i$,
    \begin{equation}\label{eq:block}
        \sum_{j\in M_k} w_j C_j = 4X_k\sum_{i<k} i X_i + 2k X_k^2 + 2M X_k - 2B k(2n+1-k) X_k .
    \end{equation}

    Summing \eqref{eq:block} over $k=1,\dots,n$ and using $\sum_k X_k = nB$ we obtain the overall objective
    \begin{equation}\label{eq:F}
        F = \sum_{k=1}^n \Bigl(4X_k\sum_{i<k} i X_i + 2k X_k^2\Bigr) + 2M nB - 2B\sum_{k=1}^n k(2n+1-k) X_k .
    \end{equation}

    The double sum in \eqref{eq:F} admits a compact representation in terms of the $\min(i,k)$ function:
    \begin{align*}
        \sum_{k=1}^n \Bigl(4X_k\sum_{i<k} i X_i + 2k X_k^2\Bigr) = 2\sum_{i=1}^n\sum_{k=1}^n \min(i,k) X_i X_k .
    \end{align*}
    For the linear term we use the identity $k(2n+1-k) = 2\sum_{i=1}^n \min(i,k)$, which is easily verified. Substituting these two identities into \eqref{eq:F} gives
    \begin{equation}\label{eq:F2}
        F = 2\sum_{i,k} \min(i,k) X_i X_k - 4B\sum_{i,k} \min(i,k) X_k + 2M nB .
    \end{equation}

    Now define $y_i = X_i - B$ for $i=1,\dots,n$. The constraint $\sum_i X_i = nB$ implies $\sum_i y_i = 0$. Replacing $X_i$ by $y_i+B$ in \eqref{eq:F2} and expanding,
    \begin{align*}
        F & = 2\sum_{i,k} \min(i,k) (y_i+B)(y_k+B) - 4B\sum_{i,k} \min(i,k) (y_k+B) + 2M nB \\ & = 2\sum_{i,k} \min(i,k) y_i y_k + 4B\sum_{i,k} \min(i,k) y_i + 2B^2\sum_{i,k} \min(i,k) \\ & \qquad - 4B\sum_{i,k} \min(i,k) y_k - 4B^2\sum_{i,k} \min(i,k) + 2M nB .
    \end{align*}
    Because the matrix $(\min(i,k))$ is symmetric, the two linear terms cancel. Collecting the constants we obtain
    \begin{align*}
        F = 2\,Q(y) + T,
    \end{align*}
    where
    \begin{align*}
        Q(y) = \sum_{i=1}^n\sum_{k=1}^n \min(i,k) y_i y_k,\qquad T = 2M nB - 2B^2\sum_{i=1}^n\sum_{k=1}^n \min(i,k).
    \end{align*}

    The quadratic form $Q(y)$ is {positive definite}. Indeed, by using the representation $\min(i,k) = \sum_{r=1}^n \mathbf{1}_{\{r\le i\}}
        \mathbf{1}_{\{r\le k\}}$ and exchanging the order of summation we get
    \begin{align*}
        Q(y) = \sum_{r=1}^n {\Bigl(\sum_{i=r}^n y_i\Bigr)}^2 \ge 0 .
    \end{align*}
    If $Q(y)=0$, then each partial sum $\sum_{i=r}^n y_i$ must be zero, which forces $y_r = 0$ for every $r$. Hence $F \ge T$, and the minimum value $T$ is achieved exactly when $X_k = B$ for all $k$.

    We can now complete the reduction.  If the \textsc{3-Partition} instance has a solution, assign the three jobs of each triple to a distinct mode; then $X_k = B$ for every $k$, and the block schedule attains $F = T$.  Conversely, suppose there exists a schedule with $F \le T$.  Since $T$ is a lower bound, we must have $F = T$, which forces $X_k = B$ for all $k$.  Each $X_k$ is a sum of distinct integers $s_j$ satisfying $\frac{B}{4} < s_j < \frac{B}{2}$.  A sum equal to $B$ can be formed only by exactly three such numbers: two numbers would sum to at most $B$, but the strict inequality $s_j < B/2$ makes the sum strictly less than $B$; four numbers would sum to at least $B$, with $s_j > B/4$ forcing a strict inequality $>B$.  Thus each $M_k$ consists of precisely three jobs whose $s_j$ values sum to $B$, and the sets $M_1,\dots,M_n$ form a valid \textsc{3-Partition}.

    We have polynomially reduced \textsc{3-Partition} to the decision version of our scheduling problem. Because \textsc{3-Partition} is strongly NP-complete, $1|\Mode(p_{ij},q_{ij})|\sum w_j C_j$ is strongly NP-hard.
\end{proof}

\subsection{Problem $1 | \Mode(p_{ij}, q_{ij})| C_{\min}$}

Finally, we consider maximizing $C_{\min}$, whose $\delta = 2$ hardness was established in
\Cref{thm:d2Cmin}. We show that, as with $C_{\max}$, lifting the restriction on $\delta$ {upgrades
this hardness to strong NP-hardness}.

\begin{theorem}\label{thm:Cmin-strong}
    The problem $1 | \Mode(p_{ij}, q_{ij})| C_{\min}$ is strongly NP-hard.
\end{theorem}

\begin{proof}
    We give a polynomial-time reduction from \textsc{3-Partition}. Set $M^* = \frac{n(n+1)}{2}B$ and, for $k=1,\dots,n$, define
    \begin{align*}
        M_k = M^* - \frac{k(k-1)}{2}B .
    \end{align*}
    Observe that $M_1 = M^*$, $M_n = nB$, and $M_k - M_{k+1} = kB > 0$, so the $M_k$ are strictly decreasing. We now construct the scheduling instance.
    \begin{itemize}
        \item \text{Ordinary jobs} $J_1,\dots,J_{3n}$: each $J_j$ has $n$ modes.  In mode $k$ ($1\le k\le n$),
              \begin{align*}
                  p_{kj} = k s_j,\qquad q_{kj} = M_k - k s_j .
              \end{align*}
              Since $s_j < B/2$ and one checks that $M_k \ge k B/2$, all $q_{kj}$ are non-negative.
        \item \text{Anchor jobs} $A_1,\dots,A_n$: $A_k$ has a single mode with $p_{A_k}=0$ and $q_{A_k}=M_k$.
        \item \text{Final job} $F$: a single mode with $p_F = q_F = 0$.
    \end{itemize}
    The instance contains $4n+1$ jobs, each ordinary job has $n$ modes, and every numerical value is bounded by $O(n^2B)$; hence the reduction is polynomial.

    {For any fixed mode assignment, \Cref{lem:d1Cmin} shows that there exists an optimal sequence ordering the jobs by non-increasing values of $p_j+q_j$.} For our instance,
    \begin{align*}
        p_{A_k}+q_{A_k} = M_k,\qquad p_{kj}+q_{kj} = M_k,\qquad p_F+q_F = 0 .
    \end{align*}
    Since $M_1 > M_2 > \cdots > M_n > 0$, {we may restrict attention to schedules that} group all jobs with $p_j+q_j = M_1$ first, then those with $M_2$, and so on, finishing with $F$. Consequently the schedule splits into $n$ contiguous blocks. Block $k$ consists of the anchor $A_k$ and all ordinary jobs assigned to mode $k$. {Within block $k$, a job that starts when the cumulative processing time equals $P$ completes at $P + M_k$; hence the minimum completion time within the block is attained by its first job and equals $S_k + M_k$, where $S_k$ is the start time of the block. The role of the anchor jobs is precisely to guarantee that every block $k$ is nonempty and therefore contributes this constraint, even if no ordinary job is assigned to mode $k$.}

    Let $\mathcal{J}_k$ be the set of ordinary jobs assigned to mode $k$ and define
    \begin{align*}
        X_k = \sum_{j \in \mathcal{J}_k} s_j .
    \end{align*}
    Since every ordinary job receives exactly one mode,
    \begin{equation}\label{eq:Cmin-sumX}
        \sum_{k=1}^{n} X_k = nB .
    \end{equation}
    The total processing time of block $k$ is $x_k = \sum_{j\in\mathcal{J}_k} p_{kj} = k X_k$, so the block start times are
    \begin{align*}
        S_1 = 0,\qquad S_k = \sum_{i=1}^{k-1} i X_i \ \ (k=2,\dots,n),\qquad S_{n+1} = \sum_{i=1}^{n} i X_i .
    \end{align*}
    The minimum completion time of block $k$ is
    \begin{align*}
        C_1^{\min} = M^*,\qquad C_k^{\min} = S_k + M_k = \sum_{i=1}^{k-1} i X_i + M^* - \frac{k(k-1)}{2}B \ \ (k\ge 2).
    \end{align*}
    The final job $F$ completes at $C_F = S_{n+1} = \sum_{i=1}^{n} i X_i$. Thus the overall minimum is
    \begin{align*}
        \min_j C_j = \min\bigl\{M^*, C_2^{\min},\dots, C_n^{\min}, C_F\bigr\}.
    \end{align*}

    The scheduling instance asks whether there exists a schedule with $C_{\min} \ge M^*$, i.e., with $\min_j C_j \ge M^*$. This is equivalent to requiring $C_k^{\min} \ge M^*$ for every $k\ge 2$ and $C_F \ge M^*$. Writing these inequalities in terms of the $X_i$ gives
    \begin{align*}
        \sum_{i=1}^{m} i X_i \ge \frac{m(m+1)}{2}B \qquad\text{for } m=1,\dots,n .
    \end{align*}
    Now introduce the deviations $Z_i = X_i - B$. From \eqref{eq:Cmin-sumX} we have $\sum_{i=1}^{n} Z_i = 0$, and the inequalities become
    \begin{align*}
        \zeta_m := \sum_{i=1}^{m} i Z_i \ge 0 \qquad (m=1,\dots,n).
    \end{align*}
    Set $\zeta_0 = 0$, so $i Z_i = \zeta_i - \zeta_{i-1}$.

    We apply summation by parts (the Abel transformation) to the zero sum $\sum_{i=1}^n Z_i$:
    \begin{align*}
        0 = \sum_{i=1}^{n} Z_i = \sum_{i=1}^{n} \frac{i Z_i}{i} = \sum_{i=1}^{n} \frac{\zeta_i - \zeta_{i-1}}{i} = \frac{\zeta_n}{n} + \sum_{i=1}^{n-1} \zeta_i\!\left(\frac{1}{i} - \frac{1}{i+1}\right).
    \end{align*}
    All coefficients $\frac{1}{i}-\frac{1}{i+1}$ are strictly positive, and the $\zeta_i$ are non-negative by construction. Therefore every term on the right-hand side must be zero, which forces $\zeta_i = 0$ for all $i$. Consequently $i Z_i = \zeta_i - \zeta_{i-1} = 0$, so $Z_i = 0$ and
    \begin{align*}
        X_i = B \qquad\text{for every } i=1,\dots,n.
    \end{align*}

    Each $X_i$ is a sum of distinct integers $s_j$ satisfying $B/4 < s_j < B/2$.  A sum equal to $B$ can only be formed by exactly three such numbers: two would give a sum $< B$, and four would give $> B$. Hence every $\mathcal{J}_i$ contains precisely three ordinary jobs whose $s_j$ values sum to $B$, which constitutes a valid \textsc{3-Partition} solution.

    Conversely, if the \textsc{3-Partition} instance has a solution with triples $T_1,\dots,T_n$ each summing to $B$, assign the three jobs of $T_k$ to mode $k$. Then $X_k = B$ for all $k$, and a direct substitution shows $C_k^{\min} = M^*$ for every $k$ and $C_F = M^*$. Thus $\min_j C_j = M^*$, and the constructed schedule attains the target.

    We have reduced \textsc{3-Partition} polynomially to the decision version of our scheduling problem. As \textsc{3-Partition} is strongly NP-complete, the problem $1|\Mode(p_{ij},q_{ij})|C_{\min}$ is strongly NP-hard.
\end{proof}

\section{Conclusion}
\label{sec:conclusions}

This paper has introduced and analyzed a scheduling model in which each job is executed in one of
several operational modes, with each mode jointly determining the job's machine processing time and
its subsequent resource-free tail time. The model is directly motivated by the production of
programmable materials such as blue-phase liquid crystal displays, where product quality matures
during the tail phase, and by prefabricated concrete production, where the curing time depends on
the chosen mix design.

Our findings establish a comprehensive complexity landscape for this problem class. We have demonstrated that while traditional single-mode scheduling is often straightforward, the introduction of multiple operational modes significantly complicates the search for optimal schedules. Specifically, we have proved that minimizing the makespan, the range of completion times, and the total weighted completion time becomes NP-hard with as few as two modes and strongly NP-hard when the number of modes is unrestricted. On the other hand, we identified that most objectives remain polynomially solvable if the processing sequence is fixed, providing a viable path for practical applications where the production order is predetermined.

The results have immediate relevance for systems like blue-phase liquid crystal printing, where synchronizing completion times is essential for product uniformity. Our analysis of the range objective $D_{\max}$ shows that achieving perfect synchronization is computationally difficult even in a fixed-order environment, although it admits a pseudo-polynomial time algorithm that can be utilized for practical instances. Manufacturers can use these insights to balance the trade-off between higher-resolution printing modes, which entail longer tails, and overall system throughput.

\medskip
Several promising avenues for future study remain:
\begin{itemize}
    \item Approximation algorithms: Given the NP-hardness of most variants, developing
        near-optimal heuristics or algorithms with guaranteed performance bounds would be of high
        practical value.
    \item Online scheduling: In many manufacturing environments, jobs arrive over time rather
        than being known in advance. Extending this model to an online setting where mode
        selection must be made in real-time is a natural next step.
    \item Multi-machine environments: This study focused on a single-machine bottleneck. Future
        research could explore parallel-machine or flow-shop configurations to better reflect
        complex industrial assembly lines.
\end{itemize}

By providing both the theoretical limits and algorithmic possibilities of multi-mode scheduling with tails, this work offers a foundation for more efficient and synchronized advanced manufacturing processes.

\section*{Acknowledgments}

Xiandong Zhang acknowledges the financial support from the National Natural Science Foundation of China under Grant No. 71971065 and No. 72232002.

\bibliography{Ref_Operational_Printing}

\end{document}